\documentclass[a4paper]{article}

\usepackage{fullpage}

\usepackage[all=normal,bibliography=tight]{savetrees}

\usepackage{tikz}
\usetikzlibrary{snakes}
\usepackage{xspace}
\usepackage{comment}

\usepackage{amsmath,amssymb,amsthm}

\usepackage{fancyvrb, url}
\RecustomVerbatimCommand{\VerbatimInput}{VerbatimInput}%
{fontsize=\scriptsize}

\newtheorem{theorem}{Theorem}
\newtheorem{lemma}[theorem]{Lemma}
\newtheorem{corollary}[theorem]{Corollary}

\theoremstyle{definition}

\newtheorem{brrule}{Rule}

\begin{document}

\newcommand{\cvd}{\textsc{ClusterVD}}
\newcommand{\cvdname}{\textsc{Cluster Vertex Deletion}}
\newcommand{\finalconst}{1.9102}
\newcommand{\Oh}{\mathcal{O}}

\newcommand{\defproblem}[4]{
%  \vspace{1mm}
%  \hline
  \vspace{2mm}
\noindent\fbox{
  \begin{minipage}{0.96\textwidth}
  \begin{tabular*}{\textwidth}{@{\extracolsep{\fill}}lr} #1 & {\bf{Parameter:}} #3 \\ \end{tabular*}
  {\bf{Input:}} #2  \\
  {\bf{Question:}} #4
  \end{minipage}
  }
%  \vspace{1mm}
%  \hline
  \vspace{2mm}
}
\newcommand{\defnoparamproblem}[3]{
%  \vspace{1mm}
%  \hline
  \vspace{2mm}
\noindent\fbox{
  \begin{minipage}{0.96\textwidth}
  #1
  {\bf{Input:}} #2  \\
  {\bf{Question:}} #3
  \end{minipage}
  }
%  \vspace{1mm}
%  \hline
  \vspace{2mm}
}

\date{}

\title{Fast branching algorithm for \textsc{Cluster Vertex Deletion}\thanks{Partially supported by NCN grant N206567140 and Foundation for Polish Science.}}

\author{
  Anudhyan Boral\thanks{Chennai Mathematical Institute, Chennai, India, \texttt{anudhyan@cmi.ac.in}}
  \and 
  Marek Cygan\thanks{Institute of Informatics, University of Warsaw, Poland, \texttt{cygan@mimuw.edu.pl}}
  \and
  Tomasz Kociumaka\thanks{Institute of Informatics, University of Warsaw, Poland, \texttt{kociumaka@mimuw.edu.pl}}
  \and
  Marcin Pilipczuk\thanks{Institute of Informatics, University of Warsaw, Poland, \texttt{malcin@mimuw.edu.pl}}
}

\maketitle

\begin{abstract}
In the family of clustering problems, we are given a set
of objects (vertices of the graph), together with some
observed pairwise similarities (edges). The goal is to identify clusters
of similar objects by slightly modifying the graph to obtain
a cluster graph (disjoint union of cliques).

H\"{u}ffner et al.~[Theory Comput. Syst. 2010] initiated the parameterized study
of \cvdname{}, where the allowed modification is vertex deletion,
and presented an elegant $\Oh(2^k k^9 + nm)$-time fixed-parameter algorithm, parameterized
by the solution size. In our work, we pick up this line of research
and present an $\Oh(\finalconst^k (n+m))$-time branching algorithm.
\end{abstract}

\section{Introduction}

The problem to cluster objects based on their pairwise similarities
has arisen from applications both in computational biology~\cite{Ben-DorSY99}
and machine learning~\cite{Bansal04}.
In the language of graph theory, as an input we are given a graph where
vertices correspond to objects, and two objects are connected by an edge
if they are observed to be similar. The goal is to transform the graph
into a cluster graph (a disjoint union of cliques) using a minimum number
of modifications.

The set of allowed modifications depends on a particular problem
and an application considered. Probably the most studied variant is the
\textsc{Cluster Editing} problem, known also as \textsc{Correlation Clustering},
where we seek for a minimal number of edge editions to obtain a cluster graph.
The study of \textsc{Cluster Editing} include%
%~\cite{AilonCN08,AlonMMN05,AroraBKSH05,Bansal04,CharikarGW05j,CharikarW04,GiotisG06,ShamirST04}
~\cite{AilonCN08,AlonMMN05,CharikarW04,GiotisG06,ShamirST04}
and, from the parameterized perspective,%
~\cite{bocker:iwoca,BockerBBT08,BockerBK11,BockerD11,BodlaenderFHMPR10,Damaschke10,FellowsGKNU11,ours-stacs,GrammGHN05,GuoKKU11,GuoKNU10,komusiewicz:thesis,komusiewicz:sofsem,ProttiSS09}.

The main principle of the parameterized complexity
is that we seek for algorithms that are efficient if the considered parameter is small.
However, the distance measure in \textsc{Cluster Editing}, the number of edge editions,
may be quite large in practical instances, and, in the light of recent lower bounds refuting
the existence of subexponential FPT algorithms for \textsc{Cluster Editing}~\cite{ours-stacs,komusiewicz:thesis},
it seems reasonable to look for other distance measures (see e.g. Komusiewicz's PhD thesis~\cite{komusiewicz:thesis}) and/or different problem formulations.

In 2008, H\"{u}ffner et al.~\cite{cvd2k} initiated the parameterized study of the
\cvdname{} problem (\cvd{} for short). Here, the allowed modification is a vertex
deletion.

\defproblem{\cvdname{} (\cvd{})}{An undirected graph $G$ and an integer $k$.}{$k$}{Does there exist a set $S$ of at most $k$ vertices of $G$ such that $G \setminus S$ is a cluster graph, i.e., a disjoint union of cliques?}

In terms of motivation, we want to refute as few objects as possible to make the
set of observations completely consistent. As a vertex deletion removes as well
all its incident edges, we may expect that this new editing measure may be significantly
smaller in practical applications than the edge-edition distance.

As \cvd{} can be equivalently stated as the problem of hitting, with minimum number of vertices, all induced $P_3$s (paths on $3$ vertices) in the input graph, \cvd{} can
be solved in $\Oh(3^k (n+m))$ time by a straightforward branching algorithm~\cite{cai},
where $n$ and $m$ denote the number of vertices and edges of $G$, respectively.
The dependency on $k$ can be improved by considering more elaborate case distinction
in the branching algorithm, either directly~\cite{gramm:cvd}, or via a general
algorithm for \textsc{3-Hitting Set}~\cite{magnus:thesis}.
H\"{u}ffner et al.~\cite{cvd2k} provided an elegant $\Oh(2^kk^9 +nm)$-time algorithm,
using the iterative compression principle~\cite{reed:ic} and a reduction
to the weighted maximum matching problem.

In our work we pick up this line of research and obtain
the fastest algorithm for (unweighted) \cvd{}.
%claim the lead in the FPT race for the fastest algorithm for (unweighted) \cvd{}.

\begin{theorem}\label{thm:main}
\cvdname{} can be solved in $\Oh(\finalconst^k (n+m))$ time and polynomial space
on an input $(G,k)$ with $|V(G)| = n$ and $|E(G)| = m$.
\end{theorem}

Contrary to the algorithm of~\cite{cvd2k}, our algorithm is a typical branching
algorithm, where a number of branches and reductions is presented, and the complexity
is analysed through (sometimes long) case analysis and branching vectors.
The advantage of this approach
is that we obtain a linear dependency on the graph size in the running time.

The main observation in the proof of Theorem~\ref{thm:main} is that,
if, for some vertex $v$, we know that there exists a solution $S$ not containing $v$,
in the neighbourhood of $v$ the \cvd{} problem reduces to \textsc{Vertex Cover}.
More precisely, define $N_1$ and $N_2$ to be the vertices within distance $1$ and $2$
from $v$, respectively, and define
the auxiliary graph $H_v$ to be a graph on $N_1 \cup N_2$
having and edge for each edge of $G$ between $N_1$ and $N_2$ and for each
non-edge inside $N_1$
in $G$.
In other words, two vertices are connected by an edge in $H_v$ iff, together with $v$, they form a $P_3$ in $G$.
We observe that a solution $S$ not containing $v$
needs to contain a vertex cover of $H_v$.
Moreover, one can show that we may greedily take
as much as possible (inclusion-wise) vertices from $N_2$ into the aforementioned
vertex cover, as these vertices would help us resolve the remaining part of the graph.

We note that a similar observation has been already used in~\cite{cvd2k} to cope
with a variant of \cvd{} where we restrict the number of clusters in the resulting
graph.

Branching to find the `correct' vertex cover of $H_v$ is a very efficient branching,
with worst-case $(1,2)$ (i.e., golden-ratio) branching vector. However,
we do not have the vertex $v$ beforehand, and branching to obtain such a vertex
may be quite costly. Thus, our approach is to get as much gain as possible
from the vertex cover-style branching on the auxiliary graph $H_v$, to be able
to balance the loss from some inefficient branches used to obtain the vertex $v$
to start with. Consequently, we employ quite involved
analysis of properties and branching algorithms for the auxiliary graph $H_v$.

The paper is organised as follows. We give some preliminary definitions
and notation in Section~\ref{sec:prelims}.
In Section~\ref{sec:Hv} we analyse the auxiliary graph $H_v$ and show
a branching algorithm finding all relevant vertex covers of $H_v$.
Then, in Section~\ref{sec:alg} we prove Theorem~\ref{thm:main}.
Section~\ref{sec:conc} concludes the paper.

\newcommand{\minvc}{\texttt{MinVC}}

\section{Preliminaries}\label{sec:prelims}

We use standard graph notation. All our graphs are undirected and simple.
For a graph $G$, by $V(G)$ and $E(G)$ we denote its vertex- and edge-set, respectively.
For $v \in V(G)$, the set $N_G(v) = \{u| uv \in E(G)\}$ is the neighbourhood
of $v$ in $G$ and $N_G[v] = N_G(v) \cup \{v\}$ is the closed neighbourhood.
We extend these notions to sets of vertices $X \subseteq V(G)$
by $N_G[X] = \bigcup_{v \in X} N_G[v]$ and $N_G(X) = N_G[X] \setminus X$.
We omit the subscript if it is clear from the context.
For a set $X \subseteq V(G)$ we also define $G[X]$ to be the subgraph induced by $X$
and $G \setminus X$ is a shorthand for $G[V(G) \setminus X]$.
A set $X \subseteq V(G)$ is called a {\em{vertex cover}} of $G$ if
$G \setminus X$ is edgeless.
By $\minvc(G)$ we denote the size of the minimum vertex cover of $G$.

In all further sections, we assume we are given an instance $(G,k)$ of
\cvdname{}, where $G = (V,E)$. That is, we use $V$ and $E$ to denote the vertex-
and edge-set of the input instance $G$.

A $P_3$ is an ordered set of $3$ vertices
$(u,v,w)$ such that $uv, vw \in E$ and $uw \notin E$.
A graph is a cluster graph iff it does not contain any $P_3$; hence, in \cvd{}
we seek for a set of at most $k$ vertices that hits all $P_3$s.

If at some point a vertex $v$ is fixed in the graph $G$, we define sets $N_1 = N_1(v)$
and $N_2 = N_2(v)$ as follows: $N_1 = N_G(v)$ and $N_2 = N_G(N_G[v])$. That is,
$N_1$ and $N_2$ are sets of vertices within distance $1$ and $2$ from $v$, respectively.
For a fixed $v \in V$,
we define an auxiliary graph $H_v$ with
$V(H_v) = N_1 \cup N_2$ and
$$E(H_v) = \{uw | u,w \in N_1, uw \notin E\} \cup \{uw | u \in N_1, w \in N_2, uw \in E \}.$$
Thus, $H_v$ consists of the vertices in $N_1$ and $N_2$
along with non-edges among vertices of $N_1$ and edges between $N_1$ and $N_2$.
Observe the following.
\begin{lemma}\label{lem:Hv-P3}
For $u,w \in N_1 \cup N_2$, we have $uw \in E(H_v)$ iff
$u$, $w$ and $v$ form a $P_3$ in $G$.
\end{lemma}

\begin{proof}
For every $uw \in E(H_v)$ with $u, w \in N_1$, $(u,v,w)$ is a $P_3$ in $G$.
For $uw \in E(H_v)$ with $u \in N_1$ and $w \in N_2$, $(v,u,w)$ forms a $P_3$ in $G$.
In the other direction, for any $P_3$ in $G$ of the form $(u,v,w)$ we have
$u,w \in N_1$ and $uw \notin E$, thus $uw \in E(H_v)$.
Finally, for any $P_3$ in $G$ of the form $(v,u,w)$ we have $u \in N_1$, $w \in N_2$
and $uw \in E$, hence $uw \in E(H_v)$.
\end{proof}

We call a subset $S \subseteq V$ a \textit{modulator} when $G \setminus S$ is a cluster graph, that is, a collection of cliques.
A modulator with minimal cardinality is called a \textit{solution}.

Our algorithm is a typical branching algorithm, that is, it consists of a number
of {\em{branching steps}}.
In a step $(A_1,A_2,\ldots,A_r)$,
$A_1,A_2,\ldots,A_r \subseteq V$, we independently consider $r$ subcases.
In the $i$-th subcase we look for a solution $S$ containing $A_i$: we delete
$A_i$ from the graph and decrease the parameter $k$ by $|A_i|$.
If $k$ becomes negative, we terminate the current branch and return a negative answer
from the current subcase.
For brevity, we sometimes write in the branching step
$w$ instead of $\{w\}$ if $A_i = \{w\}$ for some $i$.

The {\em{branching vector}} for a step $(A_1,A_2,\ldots,A_r)$ is
the vector $(|A_1|,|A_2|,\ldots,|A_r|)$. It is well-known (see e.g.~\cite{fomin-kratsch})
that the number of final subcases of a branching algorithm is bounded
by $\Oh(c^k)$, where $c$ is the largest positive root
of an equation $1 = \sum_{i=1}^r x^{-|A_i|}$ among all branching steps
$(A_1,A_2,\ldots,A_r)$ in the algorithm.

\section{The auxiliary graph $H_v$}\label{sec:Hv}

In this section we investigate properties of the auxiliary graph $H_v$.
Hence, we assume that a \cvd{} input $(G,k)$ is given with $G=(V,E)$,
and a vertex $v \in V$ is fixed.
We first start with a few basic properties and then we build
on them an efficient branching algorithm for \cvd{}, if we know
there exists a solution not containing $v$.

\subsection{Basic properties}

\begin{lemma} \label{lemma_p3} Let $G$ be a connected graph which is not a clique. Then, for every $v \in V(G)$, there is a $P_3$ containing $v$.
\end{lemma}
\begin{proof} Consider $N(v)$. If there exist vertices $u, w \in N(v)$ such that $uw \notin E(G)$ then we have a $P_3$ $(u,v,w)$. Otherwise, since $N[v]$ induces a clique, we must have $w \in N(N[v])$ such that $uw \in E(G)$ for some $u \in N(v)$. Thus we have a $P_3$, $(v,u,w)$ involving $v$.
\end{proof}

\begin{lemma} \label{lemma_vc} Let $S$ be a modulator such that $v \notin S$.
Then $S$ contains a vertex cover of $H_v$.
\end{lemma}

\begin{proof}
Observe that if $S$ is a modulator, then $G \setminus S$ does not contain a $P_3$. 
By Lemma \ref{lem:Hv-P3}, if $v \notin S$, no edge may remain in $H_v \setminus S$
and the lemma follows.
\end{proof}

\begin{lemma} \label{lemma_vc_converse} Let $X$ be a vertex cover of $H_v$. Then, in $G \setminus X$, the connected component of $v$ is a clique.
\end{lemma}
\begin{proof} Suppose the connected component of $v$ in $G \setminus X$ is not a clique. Then by Lemma \ref{lemma_p3}, there is a $P_3$ involving $v$.
Such a $P_3$ is also present in $G$. However, by Lemma \ref{lem:Hv-P3},
as $X$ is a vertex cover of $H_v$, $X$ intersects such a $P_3$, a contradiction.
\end{proof}

\begin{lemma} \label{prefer_n2} Let $S$ be a modulator such that $v \notin S$. Denote by $X$ the set $S \cap V(H_v)$. Let $Y$ be a vertex cover of $H_v$. Suppose that $X \cap N_2 \subseteq Y \cap N_2$. Then $T \triangleq (S \setminus X) \cup Y$ is also a modulator.
\end{lemma}
\begin{proof}
Since $Y$ (and hence, $T \cap V(H_v)$) is a vertex cover of $H_v$ and $v \notin T$, we know by Lemma \ref{lemma_vc_converse} that the connected component of $v$ in $G \setminus T$ is a clique. If $T$ is not a modulator, then there must be a $P_3$ contained in $Z \setminus T$, where $Z = V \setminus (\{v\} \cup N_1)$. But since $S \cap Z \subseteq T \cap Z$, $G \setminus S$ would also contain such a $P_3$.
\end{proof}

For vertex covers of $H_v$, $X$ and $Y$ , we say $X$ \textit{dominates} $Y$ if $|X| \leq |Y|$,
$X \cap N_2 \supseteq Y \cap N_2$ and at least one of these inequalities is sharp.
Two vertex covers $X$ and $Y$ are said to be \textit{equivalent} if $X \cap N_2 = Y \cap N_2$ and $|X \cap N_1| = |Y \cap N_1|$.
We note that the first aforementioned relation is transitive and strongly anti-symmetric, whereas the second
is an equivalence relation.

As a corollary of Lemma \ref{prefer_n2}, we have: 
\begin{corollary}\label{cor:dom}
Let $S$ be a modulator such that $v \notin S$.
Suppose $Y$ is a vertex cover of $H_v$ which either dominates or is equivalent to the vertex cover $X = S \cap V(H_v)$.
Then $T \triangleq (S \setminus X) \cup Y$ is also a modulator with $|T| \leq |S|$.
\end{corollary}

\subsection{Branching algorithm}\label{sec:Hv:br}

We are now ready to develop a branching algorithm that guesses the `correct'
vertex cover of $H_v$. Recall that we are working in the setting where
we look for a solution to \cvd{} on $(G,k)$ not containing $v$, thus,
by Lemma \ref{lemma_vc}, containing a vertex cover of $H_v$.
Our goal is to branch into a number of subcases, in each subcase
picking a vertex cover of $H_v$.
By Corollary \ref{cor:dom}, our branching algorithm,
to be correct, needs only to generate at least one element from each
equivalence class of the `equivalent' relation,
among maximal elements in the `dominate' relation.

The algorithm consists of a number of branching steps; in each subcase
of each step we take a number of vertices into the constructed vertex cover of $H_v$
and, consequently, into the constructed solution to \cvd{} on $G$.
At any point, the first applicable rule is applied.

First, we disregard isolated vertices in $H_v$.
Second, we take care of large-degree vertices.
\begin{brrule}\label{r:Hv:deg3}
If there is a vertex $u \in V(H_v)$ with degree at least $3$ in $H_v$,
include either $u$ or $N_{H_v}(u)$ into the vertex cover.
That is, use the branching step $(u,N_{H_v}(u))$.
\end{brrule}
Note that Rule~\ref{r:Hv:deg3} yields a branching vector $(1,d)$, where
$d \geq 3$ is the degree of $u$ in $H_v$.

Henceforth, we can assume that vertices have degree $1$ or $2$ in $H_v$.
Assume there exists $u \in N_1$ of degree $1$, with $uw \in E(H_v)$.
Moreover, assume there exists a solution $S$ containing $u$.
If $w \in S$, then, by Lemma \ref{prefer_n2}, $S \setminus \{u\}$ is also a modulator,
a contradiction.
Otherwise, if $w \in N_2$, then $(S \setminus \{u\}) \cup \{w\}$ dominates $S$.
Finally, if $w \in N_1$, then $(S \setminus \{u\}) \cup \{w\}$ is equivalent to $S$.
Hence, we infer the following greedy rule.
\begin{brrule}\label{r:Hv:deg1n1}
If there is a vertex $u \in N_1$ of degree $1$ in $H_v$, include $N_{H_v}(u)$
into the vertex cover. That is, use the branching step $(N_{H_v}(u))$.
\end{brrule}

Now we assume vertices in $N_1$ are of degree exactly $2$ in $H_v$.
Suppose we have vertices $u, w \in N_1$ with $uw \in E(H_v)$.
We would like to branch on $u$ as in Rule~\ref{r:Hv:deg3}, including
either $u$ or $N_{H_v}(u)$ into the vertex cover.
However, note that in the case where $u$ is deleted, Rule~\ref{r:Hv:deg1n1} 
is triggered on $w$ and consequently the other neighbour of $w$ is deleted.
Hence, we infer the following rule.
\begin{brrule}\label{r:Hv:22}
If there are vertices $u,w \in N_1$, $uw \in E(H_v)$ then include
either $N_{H_v}(w)$ or $N_{H_v}(u)$ into the vertex cover.
That is, use the branching step $(N_{H_v}(w), N_{H_v}(u))$.
\end{brrule}
Note that Rule~\ref{r:Hv:22} yields the branching vector $(2,2)$.

We are left with the case where the maximum degree of $H_v$ is $2$,
there are no edges with both endpoints in $N_1$,
and no vertices of degree one in $N_1$.
Hence $H_v$ must be a collection of even cycles and paths 
(recall that $N_2$ is an independent set in $H_v$).
On each such cycle $C$, of $2l$ vertices, the vertices of $N_1$ and $N_2$ alternate.
Note that we must use at least $l$ vertices for the vertex cover of $C$.
By Lemma \ref{prefer_n2} it is optimal to greedily select the $l$ vertices in $C \cap N_2$.
\begin{brrule}\label{r:Hv:even-cycle}
If there is an even cycle $C$ in $H_v$ with every second vertex in $N_2$,
include $C \cap N_2$ into the vertex cover.
That is, use the branching step $(C \cap N_2)$.
\end{brrule}
For an even path $P$ of length $2l$, we have two choices.
If we are allowed to use $l + 1$ vertices in the vertex cover of $P$, then, by Lemma~\ref{prefer_n2}, we may greedily take $P \cap N_2$.
If we may use only $l$ vertices, the minimum possible number, we need to choose
$P \cap N_1$, as it is the unique vertex cover of size $l$ of such path.
Hence, we have an $(l,l+1)$ branch with our last rule.
\begin{brrule}\label{r:Hv:last}
Take the longest possible even path $P$ in $H_v$ and either
include $P \cap N_1$ or $P \cap N_2$ into the vertex cover.
That is, use the branching step $(P \cap N_1,P \cap N_2)$.
\end{brrule}
In Rule~\ref{r:Hv:last}, we pick the longest possible path to avoid 
the branching vector $(1,2)$ as long as possible; this is the worst branching vector
in the algorithm of this section.

When we are forced to use the $(1,2)$ branch, we exploit a very specific structure
of $H_v$.
A {\em{seagull}} is a connected component of $H_v$ that is isomorphic to a $P_3$
with middle vertex in $N_1$ and endpoints in $N_2$.
The graph $H_v$ is called an {\em{$s$-skein}} if it is a disjoint
union of $s$ seagulls and some isolated vertices.
The following observation is straightforward from the above
analysis.
\begin{lemma}\label{lem:skein}
If the algorithm of Section~\ref{sec:Hv:br} may only 
use a branch with the branching vector $(1,2)$, 
then $H_v$ is an $s$-skein for some $s \geq 1$.
\end{lemma}

We conclude this section with a note on how fast a single branching step may be executed.
Note that, as $H_v$ contains parts of the complement of $G$, it may have size superlinear
in the size of $G$. However, it is easy to see that the following oracle procedure
suffices to find and execute the lowest-numbered available branching step in the graph $H_v$.
\begin{lemma}\label{lem:Hv:linear}
Given a designated vertex $v \in V$, one can in linear time either compute a vertex $w$ of degree at least
$3$ in $H_v$, together with its neighbourhood in $H_v$, or explicitely construct the graph $H_v$.
\end{lemma}
\begin{proof}
First, mark vertices of $N_1$ and $N_2$. Second, for each vertex of $G$ compute its number of neighbours
in $N_1$ and $N_2$. This information, together with $|N_1|$, suffices to compute degrees of vertices
in $H_v$. Hence, we may identify a vertex of degree at least $3$ in $H_v$, if it exists.
For such a vertex $w$, computing $N_{H_v}(w)$ takes time linear in the size of $G$.
If no such vertex $w$ exists, the complement of $G[N_1]$ has size linear in $|N_1|$ and we may
construct $H_v$ in linear time in a straightforward manner.
\end{proof}

\section{Algorithm}\label{sec:alg}

In this section we show our algorithm for \cvd{}, proving Theorem~\ref{thm:main}.
The algorithm is a typical branching algorithm, where at each step
we choose one branching rule and apply it. In each subcase, a number of vertices
is deleted, and the parameter $k$ drops by this number.
If $k$ becomes negative, the current subcase is terminated with a negative answer.
On the other hand, if $k$ is nonnegative and $G$ is a cluster graph,
the vertices deleted in this subcase form a modulator of size at most $k$.

\subsection{Preprocessing}

At each step, we first preprocess simple connected components of $G$.

\begin{lemma}\label{lem:preprocess}
In linear time,
we can for each connected component $C$ of $G$:
    \begin{enumerate}
    \item conclude that $C$ is a clique; or
    \item conclude that $C$ is not a clique, but identify a vertex $w$ such that $C \setminus \{w\}$ is a cluster graph; or
    \item conclude that none of the above holds.
    \end{enumerate}
\end{lemma}

\begin{proof}
On each connected component $C$, we perform a depth-first search. At every stage, we
ensure that the set of already marked vertices induces a clique.

When we enter a new vertex, $w$, adjacent to a marked vertex $v$,
we attempt to maintain this invariant.
We check if the number of marked vertices is equal to the number neighbours of
$w$ which are marked; if so then the new vertex $w$ is marked.
Since $w$ is adjacent to every marked vertex, the set of marked vertices remains a clique.
Otherwise, there is a marked vertex $u$ such that $uw \notin E(G)$, and we may discover
it by iterating once again over edges incident to $w$.
In this case, we have discovered a $P_3$ $(u,v,w)$ and $C$ is not a clique.
At least one of $u, v, w$ must be deleted to make $C$ into a cluster graph.
We delete each one of them, and repeat the algorithm (without further recursion)
to check if the remaining graph is a cluster graph.
If one of the three possibilities returns a cluster graph, then (2) holds.
Otherwise, (3) holds.

If we have marked all vertices in a component $C$ while maintaining
the invariant that marked vertices form a clique, then the current component
$C$ is a clique.
\end{proof}

For each connected component $C$ that is a clique, we disregard $C$.
For each connected component $C$ that is not a clique, but $C \setminus \{w\}$
is a cluster graph for some $w$, we may greedily delete $w$ from $G$:
we need to delete at least one vertex from $C$, and $w$ hits all $P_3$s in $C$.
Thus, henceforth we assume that for each connected component $C$ of $G$
and for each $v \in V(C)$, $C \setminus \{v\}$ is not a cluster graph. 
In other words, we assume that we need to delete at least two vertices to solve
each connected component of $G$.

\subsection{Studying $H_v$}

Once preprocessing is no longer possible, we fix an arbitrary vertex $v$ in $G$,
and let $C$ be its connected component.
Our goal is to `resolve' the neighbourhood of $v$: either decide to delete $v$,
or guess the `correct' vertex cover of $H_v$. However, if we implement
this in a straightforward manner, we do not get the time bound promised by Theorem~\ref{thm:main}. To achieve this bound, we carefully study the cases where $H_v$ has
small vertex cover or has special structure, and discover some possible
greedy decisions that can be made.

We would like to make decision depending on the size of the minimum vertex cover
of $H_v$. As $C$ is not a clique, by Lemma \ref{lemma_p3} $H_v$ contains at least one edge,
thus $\minvc(G) \geq 1$. We first note that we can make a distinction on small
vertex covers of $G$ in linear time.
\begin{lemma}\label{lem:findvc}
In linear time, we can determine whether $H_v$
has minimum vertex cover of size 1, of size 2, or of size at least 3.
Moreover, in the first two cases we can find the vertex cover in the same time bound.
\end{lemma}
\begin{proof}
We use Lemma \ref{lem:Hv:linear} on to find, in linear time,
a vertex $w$ with degree at least 3, or generate $H_v$ explicitly.

In the latter case, $H_v$ has vertices of degree at most $2$.
Then, $H_v$ consists of paths and cycles and we can find the size of the
minimum vertex cover in linear time.
We use the fact that paths with $l$ vertices require at least $\lfloor \frac{l}{2} \rfloor$ vertices, and cycles with $l$ vertices require $\lceil \frac{l}{2} \rceil$ vertices in the vertex cover.

If we find a vertex $w$ of degree at least $3$ in $H_v$, then $w$ must be in any vertex cover of size at most $2$. Otherwise, $N(w)$ must be in the vertex cover but $|N(w)| \geq 3$. We proceed to delete $w$ and restart the algorithm of Lemma \ref{lem:Hv:linear}
on the remaining graph to check if it has a vertex cover of size $0$ or $1$.
We perform at most $2$ such restarts.
Finally, if we do not find a vertex cover of size at most $2$,
it must be the case that the minimum vertex cover contains at least $3$ vertices.
\end{proof}

We now make a few important observations about $H_v$ that will enable us
to do some greedy choices in the future.
\begin{lemma}\label{lem:X-S}
Suppose $X$ is a vertex cover of $H_v$. Then there is a solution $S$ such that either $v \notin S$ or $|X \setminus S| \geq 2$. 
\end{lemma}

\begin{proof} Suppose $S$ is a solution such that $v \in S$ and $|X \setminus S| \leq 1$. Consider $T \triangleq (S \setminus \{v\}) \cup X$. Clearly, $|T| \leq |S|$. Since $T$ contains $X$, a vertex cover, by Lemma \ref{lemma_vc_converse}, the connected component of $v$ in $G \setminus T$ is a clique. Thus, there is no $P_3$ containing $v$. Since, any $P_3$ in $G \setminus T$ which does not include $v$ must also be contained in $G \setminus S$, contradicting the fact that $S$ is a modulator, we obtain that $T$ is also a modulator. Hence, $T$ is a solution.
\end{proof}

\begin{corollary}\label{cor:vc1}
If $\minvc(H_v)=1$ then there is a solution $S$ not containing $v$.
\end{corollary}
\begin{proof}
Let $X$ be a minimum vertex cover of $H_v$, and let $S$ be a solution
promised by Lemma \ref{lem:X-S} for the vertex cover $X$. Then
$v \notin S$, as $|X \setminus S| \leq |X| = 1$.
\end{proof}

\begin{lemma} \label{vc2_w}
Suppose that $C \setminus \{v\}$ is not a cluster graph, where $C$ is the connected component containing $v$. 
Suppose further that $X = \{w_1, w_2\}$ is a minimum vertex cover of $H_v$. Then in $G \setminus \{v\}$, either the connected component containing $w_1$ is not a clique, or the connected component containing $w_2$ is not a clique.
\end{lemma}
\begin{proof} Assume the contrary. Consider a component $\widehat C$ of $C \setminus \{v\}$ which is not a clique. Since $v$ must be adjacent to each connected component of $C \setminus \{v\}$, $\widehat C \cap N_1$ must be non-empty. For any $w \in \widehat C \cap N_1$, we have that $w_1, w_2 \neq w$ and $ww_1, ww_2 \notin E$, since otherwise the result follows. If $uw \in E$ with $u \in N_2$, then, as $\{w_1, w_2\}$ is a vertex cover we must have $u = w_1$ or $u = w_2$, We would then have $w_1$ or $w_2$ contained in a non-clique $\widehat C$, contradicting our assumption. Hence $uw \in E \Rightarrow u \in N_1$. Thus $\widehat C \subseteq N_1$. As $w_1$ and $w_2$ are not contained in $\widehat C$ and they cover all edges in $H_v$, $\widehat C$ must be an independent set in $H_v$. In $G \setminus \{v\}$, therefore, $\widehat C$ must be a clique, a contradiction.
\end{proof}

\begin{lemma} \label{lem:skein-fix}
Let $v \in V$. Suppose that $H_v$ is an $s$-skein. Then there is a solution $S$ such that $v \notin S$.
\end{lemma}
\begin{proof}
Let $H_v$ consist of seaguls $(x_1,y_1,z_1), (x_2, y_2, z_2), \ldots, (x_s, y_s, z_s)$. That is, the middle vertices $y_i$'s are in $N_1$, while the endpoints $x_i$'s and $z_i$'s are in $N_2$. If $s = 1$, $\{y_1\}$ is a vertex cover of $H_v$ and Corollary~\ref{cor:vc1} yields the result. Henceforth, we assume $s \geq 2$.

As $X$ consider the set $N_1$ with all the vertices
isolated in $H_v$ removed.
Clearly $X$ is a vertex cover of $H_v$,
thus we may use $X$ as in Lemma \ref{lem:X-S} and
obtain a solution $S$. If $v \notin S$ we are done, so let us assume
$|X \setminus S| \geq 2$. 
Take arbitrary $i$ such that $y_i \in X \setminus S$. 
As $|X \setminus S| \geq 2$, we may pick another $j \neq i$, $y_j \in X \setminus S$.
The crucial observation from the definition of $H_v$
is that $(y_j,y_i,x_i)$ and $(y_j,y_i,z_i)$ are $P_3$s in $G$.
As $y_i,y_j \notin S$, we have $x_i,z_i \in S$.
Hence, since the choice of $i$ was arbitrary, we infer that for each $1 \leq i \leq s$
either $y_i \in S$ or $x_i,z_i \in S$, and, consequently, $S$ contains a vertex
cover of $H_v$.
By Lemma~\ref{lemma_vc_converse}, $S \setminus \{v\}$ is also a modulator
in $G$, a contradiction.
\end{proof}

\subsection{Branching steps}\label{sec:br}

We are now ready to present the branching steps of our algorithm.
We assume the preprocessing (Lemma~\ref{lem:preprocess}) is done and
a vertex $v$ is picked.
We first run the algorithm of Lemma~\ref{lem:findvc} to determine
if $H_v$ has a small minimum vertex cover.
Second, we run the algorithm of Lemma~\ref{lem:Hv:linear} to check
if $H_v$ is not an $s$-skein for some $s$.

We consider the following cases.
\begin{enumerate}
\item\label{br:1}
\textbf{$\minvc(H_v) = 1$ or $H_v$ is an $s$-skein for some $s$.}
  Then, by Corollary~\ref{cor:vc1} and Lemma~\ref{lem:skein-fix},
  we know there exists a solution not containing $v$. Hence,
  we run the algorithm of Section~\ref{sec:Hv:br} on $H_v$.

\item\label{br:2}\textbf{$\minvc(H_v) = 2$ and $H_v$ is not a $2$-skein.}
Assume the application of Lemma~\ref{lem:findvc} returned 
a vertex cover $X = \{w_1,w_2\}$ of $H_v$.
By Lemma~\ref{lem:X-S}, 
we may branch into the following two subcases:
in the first we look for solutions containing $v$ and disjoint with $X$,
and in the second, for solutions not containing $v$.

In the first case, we first delete $v$ from the graph and decrease $k$ by one.
Then we check whether the connected component containing $w_1$ or $w_2$ is not a clique;
By Lemma~\ref{vc2_w}, for some $w \in \{w_1,w_2\}$, the connected component
of $G \setminus \{v\}$ containing $w$ is not a clique; finding such $w$ clearly takes
linear time. We invoke the algorithm of Section~\ref{sec:Hv:br} on $H_w$.

In the second case, we invoke the algorithm of Section~\ref{sec:Hv:br} on $H_v$.

\item\label{br:3}
\textbf{$\minvc(H_v) \geq 3$ and $H_v$ is not an $s$-skein for some $s \geq 3$.}
We branch into two cases: we look for a solution containing $v$ or not containing $v$.
In the first branch, we simply delete $v$ and decrease $k$ by one.
In the second branch, we invoke the algorithm of Section~\ref{sec:Hv:br}
on $H_v$.
\end{enumerate}

\subsection{Complexity analysis}

\newcommand{\brT}{\mathbb{T}}

In the previous discussion we have argued that invoking each branching step
takes linear time. As in each branch we decrease the parameter $k$ by at least
one, the depth of the recursion is at most $k$.
In this section we analyse branching vectors occuring in our algorithm.
To finish the proof of Theorem~\ref{thm:main} we need to show that the largest
positive root of the equation $1=\sum_{i=1}^r x^{-a_i}$ among
all possible branching vectors $(a_1,a_2,\ldots,a_r)$ is strictly less than $\finalconst$.

As the number of resulting branching vectors in the analysis is rather large,
we use a Python script for automated analysis (attached in the appendix).
The main reason for a large number of branching vectors is that we need to analyse
branchings on the graph $H_v$ in case when we consider $v$ not to be included
in the vertex cover. Let us now proceed with formal arguments.

In a few places, the algorithm of Section~\ref{sec:Hv:br} is invoked on the graph $H_v$
and we know that $\minvc(H_v) \geq h$ for some integer $h$.
Consider the branching tree $\brT$ of this algorithm.
For a node $x \in V(\brT)$, the {\em{depth}} of $x$ is the number
of vertices of $H_v$ deleted on the path from $x$ to the root.
We mark some nodes of $\brT$. Each node of depth less than $h$ is marked.
Moreover, if a node $x$ is of depth $d < h$ and the branching step
at node $x$ has branching vector $(1,2)$, we infer that
graph $H_v$ at this node is an $s$-skein for some $s \geq h-d$, all descendants
of $x$ in $V(\brT)$ are also nodes with branching steps with vectors $(1,2)$.
In this case, we mark all descendants of $x$ that are within distance (in $\brT$)
less than $h-d$.
Note that in this way we may mark some descendants of $x$ of depth equal
or larger than $h$.

We split the analysis of an application of the algorithm of Section~\ref{sec:Hv:br}
into two phases: the first one contains all branching steps performed on marked nodes,
and the second on the remaining nodes. In the second phase, we simply observe that
each branching step has branching vector not worse than $(1,2)$.
In the first phase, we aim to write a single branching vector summarizing the phase,
so that with its help we can balance the loss from other branches when
$v$ is deleted from the graph.

The main property of the marked nodes in $\brT$ is that their existence is granted
by the assumption $\minvc(H_v) \geq h$. That is, each leaf of $\brT$ has depth
at least $h$, and, if at some node $x$ of depth $d < h$ the graph $H_v$
is an $s$-skein, we infer that $s \geq h-d$ (as the size of minimum vertex cover
    of an $s$-skein is $s$) and the algorithm performs $s$ independent branching steps
with branching vectors $(1,2)$ in this case.
Overall, no leaf of $\brT$ is marked.

To analyse such branchings for $h=2$ and $h=3$ we employ the Python
script, supplied in the appendix. The procedure \texttt{branch\_Hv}
generates all possible branching vectors for the first branch,
assuming the algorithm of Section~\ref{sec:Hv:br} is allowed to
pick branching vectors $(1)$, $(1,3)$, $(2,2)$ or $(1,2)$
(option \texttt{allow\_skein} enables/disables the use of the $(1,2)$ vector
 in the first branch).
Note that all other vectors described in Section~\ref{sec:Hv:br}
may be simulated by applying a number of vectors $(1)$ after
one of the aforementioned branching vectors.

Let us now move to the analysis of the algorithm of Section~\ref{sec:br}.

In Case~\ref{br:1} the algorithm of Section~\ref{sec:Hv:br} performs branchings with vectors not worse
than $(1,2)$.

Consider now Case~\ref{br:2}.
If $v$ is deleted, we apply the algorithm of Section~\ref{sec:Hv:br}
to $H_w$, yielding at least one branching step (as the connected component with $w$
is not a clique). Hence, after this first branching step, we have either one subcase with
parameter drop at least $2$, or two subcases with parameter drops at least $2$
and at least $3$. Clearly, the second case yields worse branching vector.

If $v$ is not deleted, the algorithm of Section~\ref{sec:Hv:br}
is applied to $H_v$. The script invokes the procedure \texttt{branch\_Hv}
on $h = 2$ and \texttt{allow\_skein}$=$\texttt{False} to obtain a list
of possible branching vectors. For each such vector, we append entries $(2,3)$
from the subcase when $v$ is deleted.

Case~\ref{br:3}
is analysed analogously. The script invokes the procedure \texttt{branch\_Hv}
on $h = 3$ and \texttt{allow\_skein}$=$\texttt{False} to obtain a list
of possible branching vectors. For each such vector, we append the entry $(1)$
from the subcase when $v$ is deleted.

We infer that the largest root of the equation $1 = \sum_{i=1}^r x^{-a_i}$
occurs for branching vector $(1,3,3,4,4,5)$ and 
is less than $\finalconst$.
This branching vector corresponds to Case~\ref{br:3}
and the algorithm of Section~\ref{sec:Hv:br},
invoked on $H_v$, first performs a branching step with the vector $(1,3)$ and in
the branch with $1$ deleted vertex, finds $H_v$ to be a $2$-skein and performs
two independent branching steps with vectors $(1,2)$.

This analysis concludes the proof of Theorem~\ref{thm:main}.

\section{Conclusions and open problems}\label{sec:conc}

We have presented a new branching algorithm
for \cvdname{}. We hope our work will trigger a race for faster
FPT algorithms for \cvd{}, as it was in the case of the famous \textsc{Vertex Cover}
problem.

Repeating after H\"{u}ffner et al.~\cite{cvd2k}, we would like
to re-pose here the question for a linear vertex-kernel for \cvd{}.
As \cvd{} is a special case of the \textsc{3-Hitting Set} problem,
it admits an $\Oh(k^2)$-vertex kernel in the unweighted case and
an $\Oh(k^3)$-vertex kernel in the weighted one~\cite{3hs-kernel1,3hs-kernel2}.
However, \textsc{Cluster Editing} is known to admit a much smaller $2k$-vertex kernel,
  so there is a hope for a similar result for \cvd{}.

\bibliographystyle{abbrv}
\bibliography{cluster-vd}

\newpage
\appendix
\section*{Python script automating complexity analysis\footnote{
Also available at~\url{www.mimuw.edu.pl/~malcin/research/cvd}}}

\begin{Verbatim}
import scipy.optimize

def value(vector):
  """compute the value of a branching vector"""
  def h(x):
    return sum([x**(-v) for v in vector])-1
  return scipy.optimize.brenth(h,1, 100)


def join(first, then):
  """peform 'then' in each branch after the execution of 'first' """
  return [x+y for x in first for y in then]


def add(a, vector):
  """add a to each element of a vector"""
  return join([a], vector)


golden_branch = [1,2]   # golden-ratio branch, worst branch in Hv


def skein_vector(s):
  """returns branching vector from s-skein"""
  if s == 0:
    return [0]
  else:
    return join(skein_vector(s-1), golden_branch)


Hv_branches = dict()

def branch_Hv(h, allow_skein=True):
  """return list of possible branching vectors on Hv, where each subcase
  deletes at least h vertices; if allow_skein=False, ignore the case when
  Hv is a skein"""
  if h <= 0:
    return [[0]]
  # Memoize for speed-up
  if Hv_branches.has_key((h, allow_skein)):
    return Hv_branches[(h, allow_skein)]
  res = []
  # If skein is allowed, add appriopriate vector.
  if allow_skein:
    res.append(skein_vector(h))
  # Greedy step.
  # Can be applied multiple times to simulate larger drop.
  res += [add(1, v) for v in branch_Hv(h-1)]
  # Rule 1: (1,3) branch.
  # Branches (1,d) for d>3 may be simulated by subsequent greedy steps
  res += [add(1, v1) + add(3, v2) for v1 in branch_Hv(h-1) for v2 in branch_Hv(h-3)]
  # Rule 3: (2,2) branch
  res += [add(2, v1) + add(2, v2) for v1 in branch_Hv(h-2) for v2 in branch_Hv(h-2)]
  # Rule 5, if Hv is not a skein, yields (2,3) branch which can be simulated
  # by (2,2) branch + greedy step in one branch, so we omit it here.
  Hv_branches[(h, allow_skein)] = res
  return res


vectors = []  # all branching vectors

# (1,2) vector from standard branching on Hv
vectors.append(golden_branch)

# Case: MinVC(Hv) = 2, Hv is not a 2-skein
vectors += [add(1, golden_branch) + v for v in branch_Hv(2, allow_skein=False)]

# Case: MinVC(Hv) >= 3, Hv is not a skein
vectors += [[1] + v for v in branch_Hv(3, allow_skein=False)]

for v in vectors:
  print ("%.11f  : " % value(v)), v

print "Largest root: %.11f" % max([value(v) for v in vectors])

\end{Verbatim}

\end{document}